\documentclass{birkjour}
 \newtheorem{thm}{Theorem}[section]

 \theoremstyle{definition}
 \newtheorem{defn}[thm]{Definition}
 \theoremstyle{remark}

 \numberwithin{equation}{section}
\newcommand{\Pin}{\mathop{\mathrm{Pin}}}
\newcommand{\Spin}{\mathop{\mathrm{Spin}}}

\newcommand{\Dic}{\mathop{\mathrm{Dic}}}
\def\bR{\mathbb{R}}
\def\bC{\mathbb{C}}
\def\bH{\mathbb{H}}

\usepackage{lineno}
\usepackage{tikz}
\usepackage{pgfplots}
\usepackage{color}
\usepgflibrary{shapes.geometric} 
\usepgflibrary[shapes.geometric] 
\usetikzlibrary{shapes.geometric} 
\usetikzlibrary[shapes.geometric] 

\begin{document}
	\bibliographystyle{plain} 
%
%
%
%
%
%
%
%
%

\title[From the Trinity $(A_3, B_3, H_3)$ to an $ADE$ correspondence]
 {From the Trinity $(A_3, B_3, H_3)$ to an $ADE$ correspondence}

\author[P-P Dechant]{Pierre-Philippe Dechant}

\address{%
Pro-Vice Chancellor's Office,\\York St John University, York YO31 7EX, United Kingdom,\\ {}\\
York Cross-disciplinary Centre for Systems Analysis,\\
University of York, Heslington YO10 5GE, United Kingdom\\{}\\
Department of Mathematics\\
University of York, Heslington YO10 5DD, United Kingdom}

\email{ppd22@cantab.net}



\date{December 6, 2018}
\dedicatory{To the late Lady Isabel and Lord John Butterfield}


\begin{abstract}
In this paper we present novel $ADE$ correspondences by combining an earlier induction theorem of ours with one of Arnold's observations concerning Trinities,  and the McKay correspondence.
We first extend Arnold's indirect link between the  Trinity of symmetries of the Platonic solids $(A_3, B_3, H_3)$ and the Trinity of exceptional 4D root systems $(D_4, F_4, H_4)$ to an explicit Clifford algebraic construction linking the two ADE sets of root systems $(I_2(n), A_1\times I_2(n), A_3, B_3, H_3)$ and $(I_2(n), I_2(n)\times I_2(n), D_4, F_4, H_4)$. 
The latter  are connected through the McKay correspondence with the ADE Lie algebras $(A_n, D_n, E_6, E_7, E_8)$. We show that there are also novel indirect as well as direct connections between these ADE root systems and the new ADE set of root systems $(I_2(n), A_1\times I_2(n), A_3, B_3, H_3)$, resulting in a web of three-way ADE correspondences between three ADE sets of root systems.

\end{abstract}
\maketitle

\section{Introduction}
The normed division algebras -- the real numbers $\bR$, the complex numbers $\bC$ and the quaternions $\bH$ -- naturally form a unit of three: $(\bR, \bC, \bH)$. This straightforwardly extends to the associated projective spaces $(\mathbb{R}P^n, \mathbb{C}P^n, \mathbb{H}P^n)$. This includes the special case of the spheres  $(\mathbb{R}P^1=S^1, \mathbb{C}P^1=S^2, \mathbb{H}P^1=S^4)$ with their associated Hopf bundles $(S^1\rightarrow S^1, S^3\rightarrow S^2, S^7 \rightarrow S^4)$. Vladimir Arnold \cite{Arnold1999symplectization,Arnold2000AMS} noted that other objects in mathematics naturally form units of three, or that many problems over the reals have interesting complex and quaternionic generalisations. A prominent example are the Platonic solids (tetrahedron, octahedron and icosahedron) along with their symmetry groups: the rotations only are commonly denoted $(T, O, I)$ for obvious reasons and are also known as $(A_4, S_4, A_5)$ as symmetric and alternating groups; the reflection symmetry groups are denoted by $(A_3, B_3, H_3)$ in Coxeter notation (see Section \ref{sec_induct}). For the crystallographic root systems  (types $A$-$G$) this notation is the same as the usual Dynkin notation familiar from Lie algebras; but this classification also includes the non-crystallographic root systems (which do not have corresponding Lie algebras due to their non-crystallographic property): the symmetries of the regular polygons $I_2(n)$ and the exceptional root systems $H_3$ (icosahedral symmetry) and $H_4$. 

Arnold termed these triplets `Trinities'; and despite the somewhat esoteric nature and name, he used these as guiding principles for his own work: he would conjecture that there should be a complexified or quaternionified version of a theorem applying to the real numbers, which `preserves some of the essential structure', and based on this intuition he would seek to prove it. Other Trinities include the root systems $(D_4, F_4, H_4)$ in four dimensions, the Lie algebras of $E$-type $(E_6, E_7, E_8)$, or the structure of singularities. Thus, these Trinities are a useful guiding principle, both from the perspective of finding the complexification and quaternionification of a real theory or theorem, and also as regards the connections between different Trinities, which are often in quite different mathematical areas.  

This intuition must of course then be backed up by concrete proofs and constructions. In some of the above mentioned cases the connections are more obvious than others. For instance, it is easy to see the links between the Trinities clustered around  $(\bR, \bC, \bH)$. Likewise, there is a web of obvious connections between different Trinities related to the Platonic Solids such as their symmetry triples $((2,3,3), (2,3,4), (2,3,5))$ denoting the orders of the three types of rotations in each polyhedral group (e.g. the icosahedral group has 2-, 3- and 5-fold rotations),  the structure of singularities, the symmetry groups $(A_3, B_3, H_3)$, the binary polyhedral groups $(2T, 2O, 2I)$, the number of roots $(12, 18, 30)$ of the 3D root systems $(A_3, B_3, H_3)$ or the order of the binary polyhedral groups $(24, 48, 120)$. However, it is less obvious that $(A_3, B_3, H_3)$ and $(E_6, E_7, E_8)$ should be related: but the triples $((2,3,3), (2,3,4), (2,3,5))$ of orders of the three types of rotations in the polyhedral groups are also exactly the lengths of the three legs in the diagrams of $(E_6, E_7, E_8)$. Whilst this is strikingly obvious and intriguing on one level, it is not at all clear how this connection comes about concretely. 

Another connection between different areas of mathematics is the McKay correspondence \cite{Mckay1980graphs}, named after John McKay who also first noticed Moonshine \cite{gannon2006moonshine,eguchi2011notes}. It relates the binary polyhedral groups $(2T, 2O, 2I)$ and the $E$-type Lie algebras $(E_6, E_7, E_8)$ in two ways: firstly, the graphs depicting the tensor product structure of the irreducible representations of the binary polyhedral groups are exactly the graphs of the affine extensions of $(E_6, E_7, E_8)$. In this construction, each irreducible representation corresponds to a node in the diagram. Each of them is tensored with the two-dimensional spinorial irreducible representation. The rule for connecting nodes in the diagram is given by which irreducible representations occur in these tensor products. Secondly,  the sum of the dimensions of the irreducible representations (not squared) is given by $(12, 18, 30)$, which is equal to the Coxeter number $h$ (the order of the Coxeter element)  of the $E$-type Lie algebras $(E_6, E_7, E_8)$ (the sum of the dimensions squared of the irreducible representations is of course the order of the respective groups, $(24, 48, 120)$).

However, in fact this McKay correspondence is wider, and contains all subgroups of $\Spin(3)$ (which is isomorphic to the unit quaternions): the binary cyclic and dicyclic groups along with the binary tetrahedral, octahedral and icoshedral groups. The corresponding Lie algebras are those of type $A_n$ and $D_n$, making this a correspondence for $ADE$ Lie algebras, rather than just the Trinity part. 

Recently, we have shown a new explicit construction between $(A_3, B_3, H_3)$ and $(D_4, F_4, H_4)$ \cite{Dechant2012Induction, Dechant2012CoxGA}. Arnold gives his link in \cite{Arnold1999symplectization,Arnold2000AMS}. It is rather cumbersome, and involves numerous intermediate steps. Arnold himself says: ``\emph{Few years ago I had discovered an operation transforming the last trinity [$(A_3, B_3, H_3)$] into another trinity of Coxeter groups $(D_4, F_4, H_4)$. I shall describe this rather unexpected operation later.}'' If one considers the Weyl chamber Springer cone decomposition of $(A_3, B_3, H_3)$, one finds that the orders 
$$24=2(1+3+3+5)$$
$$48=2(1+5+7+11)$$
$$120=2(1+11+19+29)$$
of the groups decompose according to the number of Weyl chambers in each Springer cone with coefficients that are one less than the quasihomogeneous weights of $(D_4, F_4, H_4)$, which are $(2, 4, 4, 6)$, $(2, 6, 8, 12)$ and  $(2, 12, 20, 30)$, respectively. 

Firstly, this fails to notice that the numbers appearing in the decomposition $(1,3,3,5)$, $(1,5,7,11)$ and $(1,11,19,29)$ are just the `exponents' $m_i$ of $(D_4, F_4, H_4)$, i.e. they are related to the complex eigenvalues $\exp(2\pi i m_i/h)$ of the Coxeter element of $(D_4, F_4, H_4)$. These are well-known to be related to the degrees $d_i$ of polynomial invariants $(2, 4, 4, 6)$, $(2, 6, 8, 12)$, $(2, 12, 20, 30)$ of these groups $(D_4, F_4, H_4)$  via $d_i=m_i+1$ \cite{Humphreys1990Coxeter}. Therefore, the more direct connection is actually via the Springer cone decomposition and the exponents. 

Our recent Clifford algebraic construction \cite{Dechant2012Induction} is much more immediate and general, and rather less surprising too; furthermore, like the McKay correspondence it is a wider correspondence that encompasses the Trinity but furthermore includes countably infinite families. This construction is a statement between root systems in 3D and 4D. The theorem states that one can start with any 3D root system and  construct a corresponding 4D root system from it. So in a philosophical way it proves that these 4D root systems would have to exist; due to the accidentalness of this $3D-4D$ connection they have certain unusual properties not shared by other root systems in general. However, the list of root systems is very limited and has of course been well known for a long time. Therefore one can simply calculate the correspondence for each 3D root system explicitly. The induced 4D root systems have to be from the limited list of known 4D root systems: for the Trinity of irreducible 3D root systems $(A_3, B_3, H_3)$ one thus gets the Trinity of exceptional 4D root systems  $(D_4, F_4, H_4)$. We count $D_4$ as exceptional since it has the exceptional triality symmetry. Although the family $D_n$ of course exists in any dimension, this triality symmetry of $D_4$ of permuting the three legs (see  Fig. \ref{figE8CoxPl}) is accidental to 4D. But there is also a countably infinite family of 3D root systems $A_1\times I_2(n)$, which yields $I_2(n)\times I_2(n)$ in four dimensions (using notation for the product of the respective Coxeter groups, rather than the sum of the root systems). The case $A_1^3$ is of course contained in this, but is a simple illustrative example yielding $A_1^4$. 

There is of course a link between the 4D root systems and the even subgroups of the quaternions, and thus via the McKay correspondence to the $ADE$ Lie algebras. We note that we essentially construct the even subgroups of the quaternions as spinor groups from the 3D root systems, and that the Trinity $(12, 18, 30)$ connecting the irreducible representatios of the binary polyhedral groups and the Coxeter number of the $E$-type algebras is already the number of roots in the 3D root systems $(A_3, B_3, H_3)$. This suggests in general that one can go all the way from the 3D root systems via the binary polyhedral groups to the $ADE$ algebras, and in particular that the icosahedral root system $H_3$ should be related to $E_8$. We have recently constructed the $240$ roots of $E_8$ as a double cover of the $120$ elements of $H_3$ in the $2^3=8$-dimensional Clifford algebra of 3D space \cite{Dechant2016Birth}.

For this more general correspondence between 3D and 4D root systems it is then an interesting question to see whether Arnold's original link via the 4D Coxeter exponents carries over to all cases. We will discuss our new correspondence from this point of view in this paper. In Section \ref{sec_induct} we summarise some basic definitions and the proof of the induction correspondence via Clifford algebra. We then consider the geometry of the Coxeter plane from a Clifford algebra perspective, which allows one to completely factorise the Coxeter elements within the algebra, finding the exponents geometrically (Section \ref{sec_plane}). The 3D root systems and associated 4D Coxeter plane geometries are discussed explicitly in Section \ref{sec_cases}. In Section \ref{sec_McKay}, we consider this new set of root systems in the light of the McKay correspondence, making an indirect link between 2D/3D root systems and $ADE$, and connecting the McKay correspondence with the number of roots in the 2D/3D root systems.  
In Section \ref{sec_ADE} we consider a new direct connection between this new set of 2D/3D root systems and $ADE$ root systems, constructing the diagrams of the latter from the former. We conclude in Section \ref{sec_concl}.

\section{Background}\label{sec_induct}

The relationship between root systems, Coxeter (Weyl groups) and Lie theory is well known \cite{Humphreys1990Coxeter, FuchsSchweigert1997}. The crystallographic root systems arise in the context of root lattices in Lie theory and essentially allow one to classify semi-simple Lie algebras. The associated Coxeter groups are the Weyl groups of these root lattices. In other work, however, we were also interested in non-crystallographic groups and root systems, which do not have an associated Lie algebra, and also feature in the context of Trinities considered here \cite{DechantTwarockBoehm2011H3aff,DechantTwarockBoehm2011E8A4}. We therefore introduce the notion of a root system insofar as it is needed to prove the induction theorem. 

\subsection{Clifford algebras, Root systems and Coxeter groups}

\begin{defn}[Root system] \label{DefRootSys}
A \emph{root system} is a collection $\Phi$ of non-zero (root)  vectors $\alpha$ spanning an $n$-dimensional Euclidean vector space $V$ endowed with a positive definite bilinear form, which satisfies the  two axioms:
\begin{enumerate}
\item $\Phi$ only contains a root $\alpha$ and its negative, but no other scalar multiples: $\Phi \cap \mathbb{R}\alpha=\{-\alpha, \alpha\}\,\,\,\,\forall\,\, \alpha \in \Phi$. 
\item $\Phi$ is invariant under all reflections corresponding to root vectors in $\Phi$: $s_\alpha\Phi=\Phi \,\,\,\forall\,\, \alpha\in\Phi$. 
The reflection $s_\alpha$ in the hyperplane with normal $\alpha$ is given by $$s_\alpha: x\rightarrow s_\alpha(x)=x - 2\frac{(x|\alpha)}{(\alpha|\alpha)}\alpha,$$\label{refl} where $(\cdot \vert \cdot)$ denotes the inner product on $V$.
\end{enumerate}
\end{defn}
A subset $\Delta$ of $\Phi$, called \emph{simple roots} $\alpha_1, \dots, \alpha_n$ is sufficient to express every element of $\Phi$ via linear combinations with coefficients of the same sign. 
For a \emph{crystallographic} root system, these are $\mathbb{Z}$-linear combinations, whilst for the \emph{non-crystallographic} root systems one needs to consider certain extended integer rings. For instance  for $H_2$, $H_3$ and $H_4$ one has the extended integer ring $\mathbb{Z}[\tau]=\lbrace a+\tau b| a,b \in \mathbb{Z}\rbrace$, where $\tau$ is   the golden ratio $\tau=\frac{1}{2}(1+\sqrt{5})=2\cos{\frac{\pi}{5}}$, and $\sigma$ is its Galois conjugate $\sigma=\frac{1}{2}(1-\sqrt{5})$ (the two solutions to the quadratic equation $x^2=x+1$), and linear combinations are with respect to this $\mathbb{Z}[\tau]$. This integrality property of the crystallographic root systems  (types $A$-$G$) leads to an associated lattice which acts as a root lattice for Lie algebras, which are named accordingly. In contrast, no such lattice exists for the non-crystallographic groups (types $H$ and $I$), which accordingly do not have associated Lie algebras, and are perhaps less familiar as a result.

Root systems and simple roots are convenient for considering reflection groups: each root vector defines a hyperplane that it is normal to and therefore a reflection in that hyperplane. Multiplying together such simple reflections 
$s_i: x\rightarrow s_i(x)=x - 2\frac{(x|\alpha_i)}{(\alpha_i|\alpha_i)}\alpha_i$
 generates a reflection group. This is in fact a Coxeter group, since the simple reflections $s_i$ satisfy the defining relations:
\begin{defn}[Coxeter group] A \emph{Coxeter group} is a group generated by a set of involutory generators $s_i, s_j \in S$ subject to relations of the form $(s_is_j)^{m_{ij}}=1$ with $m_{ij}=m_{ji}\ge 2$ for $i\ne j$. 
\end{defn}

Root systems are therefore a useful paradigm for reflection groups. However, Clifford algebras \cite{Hestenes1966STA,  LasenbyDoran2003GeometricAlgebra, Porteous1995Clifford, Lounesto1997, Garling2011Clifford} are also very efficient at performing reflections and are in fact very natural -- perhaps the most natural -- objects to consider in this framework \cite{Dechant2015ICCA}: the definition of a root system only stipulated a vector space with an inner product. So without loss of generality one can construct the Clifford algebra over that vector space by using this inner product. We therefore define an algebra product via  $xy=x\cdot y+x \wedge y$, where the inner product (given by the symmetric bilinear form) is the symmetric part $x\cdot y=(x|y)=\frac{1}{2}(xy+yx)$, and the exterior product the antisymmetric part $x\wedge y=\frac{1}{2}(xy-yx)$. (This also means that parallel vectors  commute whilst orthogonal vectors anticommute.) 
We extend the algebra product via linearity and associativity. This enlarges the algebra to a $2^n$-dimensional vector space, which is isomorphic to the familiar exterior algebra, though they are not isomorphic as algebras. In fact, the Clifford algebra is much richer, since the algebra product is invertible in the sense that the  inverse of multiplication with a non-null vector $x$ is simply $x^{-1}=\frac{x}{|x|^2}$ since $xx=x\cdot x = |x|^2$ (in the positive signature spaces we will consider there are no null vectors anyway). Using this form for the inner product $x\cdot y=\frac{1}{2}(xy+yx)$ in the (simple) reflection formula $s_i: x\rightarrow s_i(x)=x-2\frac{\alpha_i\cdot x}{\alpha_i\cdot \alpha_i}  \alpha_i$ and assuming unit normalisation of roots $\alpha_i\cdot \alpha_i=1$, one gets the much simplified version for the reflection formula 
$$s_i: x\rightarrow s_i(x)=x-2\cdot \frac{1}{2}(x \alpha_i+ \alpha_ix) \alpha_i=x-x \alpha_i^2- \alpha_ix \alpha_i=-\alpha_ix\alpha_i.$$
Moreover, since via the Cartan-Dieudonn\'e theorem most `interesting' (at least from a mathematical physics perspective: orthogonal, conformal, modular) groups can be written as products of reflections \cite{Dechant2015ICCA,Hestenes1990NewFound,Hestenes2002PointGroups,Hestenes2002CrystGroups,Hitzer2010CLUCalc}, this formula actually provides a completely general way of performing such transformations by successive multiplication with the unit vectors defining the reflection hyperplanes
$$s_1\dots s_k: x\rightarrow s_1\dots s_k(x)=(-1)^k\alpha_1\dots\alpha_kx\alpha_k\dots\alpha_1=:(-1)^k Ax\tilde{A}.$$
The tilde denotes the reversal of the order of the constituent vectors in the product $A=\alpha_1\dots\alpha_k$. In order to study the groups of transformations one therefore only needs to consider products of root vectors in the Clifford algebra. This is therefore an extremely (if not completely) general way of doing group theory.

Since $\alpha_i$ and $-\alpha_i$ encode the same reflection, products of unit vectors are double covers of the respective orthogonal transformation, as $A$ and $-A$ encode the same transformation. We call even products $R$, i.e. products of an even number of vectors, spinors or rotors, and a general product $A$ versors or pinors. They form the $\Pin$ group and constitute a double cover of the orthogonal group, whilst the even products form the double cover of the special orthogonal group, called the $\Spin$ group.
Clifford algebra therefore provides a particularly natural and simple construction of the $\Spin$ groups. 

Thus the remarkably simple construction of the binary polyhedral groups (which are the spin double covers of the polyhedral groups) in our context is not at all surprising from a Clifford point of view. It was even known separately that even subgroups of the quaternions form root systems in 4D  \cite{Humphreys1990Coxeter} and that these even quaternion subgroups are in the McKay correspondence with the $ADE$ Lie algebras. But it appears that it was not appreciated that the quaternions arose in a geometric guise as the spin group in three dimensions with the even quaternion groups determined by the 3D root systems. To our knowledge the full connection had not been realised before. 

\subsection{Induction Theorem}\label{indthrm}
Whilst the above discussion was completely general for orthogonal groups in spaces of arbitrary dimension and signature (and via some isomorphisms also the conformal and modular groups \cite{Dechant2015ICCA,Dirac1936,HestenesSobczyk1984}), our construction is  based on 3D geometry, and thus very straightforward. Consider the Clifford algebra of 3D generated by three orthogonal unit vectors $e_1$, $e_2$ and $e_3$. This yields an eight-dimensional vector space generated by the elements
$$
  \underbrace{\{1\}}_{\text{1 scalar}} \,\,\ \,\,\,\underbrace{\{e_1, e_2, e_3\}}_{\text{3 vectors}} \,\,\, \,\,\, \underbrace{\{e_1e_2=Ie_3, e_2e_3=Ie_1, e_3e_1=Ie_2\}}_{\text{3 bivectors}} \,\,\, \,\,\, \underbrace{\{I\equiv e_1e_2e_3\}}_{\text{1 trivector}}.
$$
Any of the bivectors or trivectors  square to $-1$. Thus one gets different imaginary units based on a real vector space, without complexifying the whole space. See \cite{Hitzer2013Sqrts} for detail on the theory of square roots of $-1$ in Clifford algebra. In fact one needs to be careful with such imaginary units since they do not necessarily (anti)commute. In fact, the scalar and the three bivectors also satisfy quaternionic relations. Thus the other two normed division algebras emerge naturally within Clifford algebras, without the need to complexify or quaternionify the whole underlying vector space. We see in the context of the Coxeter plane that this is actually much more natural and geometrically insightful. The geometric interpretation of these elements of the Clifford algebra is still for vectors as lines or directions; and since a pair of vectors defines a plane (in any dimension) the bivectors are planes, whilst the trivector is a volume, and the scalar a point/number. Thus one sees that  despite considering the geometry of three dimensions, there is actually a natural eight-dimensional space associated with it. Furthermore, there is also a four-dimensional subalgebra, the even subalgebra consisting of the scalar and the bivectors (the one that satisfies quaternionic relations). This subalgebra is in fact the 4D space that allows us to define 4D root systems from 3D root systems. We have used the full 8D algebra in other work \cite{Dechant2016Birth, dechant2014SIGMAP}, e.g. for constructing the root system $E_8$ from the icosahedron $H_3$ or for defining representations, but will only consider the even subalgebra from now in this work.

We now have the background we need in order to prove that any 3D root system yields a 4D root system \cite{Dechant2012Induction}. Multiplying together root vectors in the Clifford algebra generally yields pinors in the full 8D algebra, but even products will stay in the even subalgebra (4D). The respective polyhedral (discrete subgroup of the special orthogonal) group acts on a vector $x$ via $\tilde{R}xR$. The spinors $R$ form a (spin) double cover of this  polyhedral group and  as well as a group under multiplication $R_1R_2$, which is actually the respective binary polyhedral group. A general spinor  has components in the even subalgebra $R=a_0+a_1e_2e_3+a_2e_3e_1+a_3e_1e_2$, which is a four-dimensional vector space. We can also endow this vector space with a Euclidean inner product by defining $(R_1,R_2)=\frac{1}{2}(R_1\tilde{R}_2+R_2\tilde{R}_1)$ for  two spinors $R_1$ and $R_2$. This induces the norm $|R|^2=R\tilde{R}=a_0^2+a_1^2+a_2^2+a_3^2$. It is thus very natural to think of spinors as living in Euclidean four-dimensional space. In fact, not only can we think of each of them as a 4D vector, but the spinor groups actually yield a collection of vectors in 4D. It is easy to show that this collection of vectors satisfies the axioms of a root system.

\begin{thm}[Induction Theorem]\label{HGA_4Drootsys}
Any 3D root system gives rise to a spinor group $G$ which induces a root system in 4D.
\end{thm}
\begin{proof}
Free multiplication of simple roots in this 3D root system yields a group of pinors in the full 8D algebra, doubly covering the Coxeter group transformations.  Even products form a spinor group $G$ in the even subalgebra, having group elements of the form $R=a_0+a_1e_2e_3+a_2e_3e_1+a_3e_1e_2$. It was shown above that this even subalgebra is a 4D vector space with an inner product $(R_1,R_2)=\frac{1}{2}(R_1\tilde{R}_2+R_2\tilde{R}_1)$. We can thus reinterpret each spinor $R$ as a 4D vector $(a_0, a_1, a_2, a_3)^T$, and we denote the collection  of the 4D vectors for all the 3D spinors  $R \in G$ by $\Phi$. These provide the vector space, inner product and collection of vectors $\Phi$ in the definition of a root system (definition \ref{DefRootSys}). It remains to check the two defining axioms of a root system for this collection $\Phi$:
\begin{enumerate}
\item By construction, $\Phi$ contains the negative of a root $R$ because both $R$ and $-R$ encode the same rotation since spinors provide a double cover. Thus if $R$ is in $\Phi$, then so is $-R$; other scalar multiples do not arise because of the normalisation to unity. 
\item Reflections with respect to the inner product $(R_1,R_2)$ defined above are given by $R_2'=R_2-2(R_1, R_2)/(R_1, {R}_1) R_1=-R_1\tilde{R}_2R_1$. But $G$ is closed under multiplication by $-1$ and reversal since $-R$ encodes the same group transformation as $R$, and $\tilde{R}$ is its inverse. Thus $-R_1\tilde{R}_2R_1\in G$ for $R_1, R_2 \in G$ by closure of the group under group multiplication, reversal and multiplication by $-1$. Therefore $\Phi$ is invariant under all reflections in the vectors $R$ and is thus a root system.
\end{enumerate}
\end{proof}

This proof does not make reference to any specific root system in 3D, and therefore allows one to construct a 4D root system for any 3D root system. This goes some way towards explaining why four dimensions are particularly rich for root systems because of additional, exceptional root systems. If we did not already know that they existed this would be a constructive proof. One can calculate each case explicitly showing how each 4D root system arises from the 3D root system (these will be discussed in detail in Section \ref{sec_cases}), but it is also obvious from the order of the 3D groups involved and the number of roots in the known 4D root systems. 
In this way $(A_3, B_3, H_3)$ gives rise to  $(D_4, F_4, H_4)$, which will turn out to be crucial in the following discussion of Trinities and ADE correspondences. The orders of the 3D Coxeter groups are $(24, 48, 120)$, which are halved by going to the even subgroup but then doubled again by going to the binary double cover. $(24, 48, 120)$ is therefore exactly the number of roots in 4D and thus they induce the root systems  $(D_4, F_4, H_4)$. The case $A_1^3$ gives $A_1^4$, and more generally the countably infinite family $A_1
\times I_2(n)$ gives rise to $I_2(n)\times I_2(n)$. We will construct these explicitly in Section \ref{sec_cases} together with the respective construction of the 4D Coxeter plane and explicit factorisation of the Coxeter element, after we have introduced these in the following section. This is motivated by the fact that Arnold's connection between $(A_3, B_3, H_3)$ and $(D_4, F_4, H_4)$ was via exponents, whilst we have in fact found a wider correspondence $(A_1\times I_2(n), A_3, B_3, H_3)\rightarrow (I_2(n)\times I_2(n), D_4, F_4, H_4)$ including the two infinite families. We will therefore extend the connection via Springer cone decomposition and exponents to this wider correspondence, establishing it more firmly as a new correspondence. We are therefore now considering the geometry of the Coxeter plane -- readers interested only in the correspondence but not in the details of the Coxeter plane may skip this next section.

\section{Coxeter element: Coxeter plane and  exponents}\label{sec_plane}

The Coxeter element $w$ of a Coxeter group is the product of all the simple reflections i.e. $w=s_1\dots s_n$. 
It has the highest order of all the elements in the group. 
This order of the Coxeter element is called the Coxeter number $h$. 
The order in the product in $w$ does matter, but all such Coxeter elements are conjugate to each other, and have equivalent descriptions in what follows. 

Picking one such Coxeter element, one can look for its eigenvectors and eigenvalues.
Since it does not in general have real eigenvalues, people usually look for complex eigenvalues by complexifying the whole space. 
This is in fact unnecessary in the Clifford algebra setting and masks the underlying geometry, as we will show in the following \cite{Dechant2016AGACSE,Dechant2012AGACSE}. 
One can show that there exists a distinguished plane in which the Coxeter element acts as an $h$-fold rotation \cite{Humphreys1990Coxeter}. 
Projection of a root system onto this Coxeter plane is thus a convenient way of visualising any finite Coxeter group in any dimension.
There are actually several such planes where the Coxeter element acts as an $h$-fold rotation, and one can look for eigenvalues of the form $\exp(2\pi i m_i/h)$. The Coxeter plane is the one distinguished by having $m=1$ and $m=h-1$. 
We will see that the complex structure in the exponential is  given by the bivector of the respective plane, that there are therefore several different complex structures and that they emerge purely algebraically from the factorisation of the Coxeter element. 
The Coxeter element therefore factorises as the product of bivector exponentials describing $h$-fold rotations in orthogonal planes. Complex eigenvalues thus arise geometrically without the need to complexify the whole real vector space. Clockwise and counterclockwise rotations in the same plane trivially yield exponents $m$ and $h-m$.

\begin{figure}
	\begin{center}

\includegraphics[width=8cm]{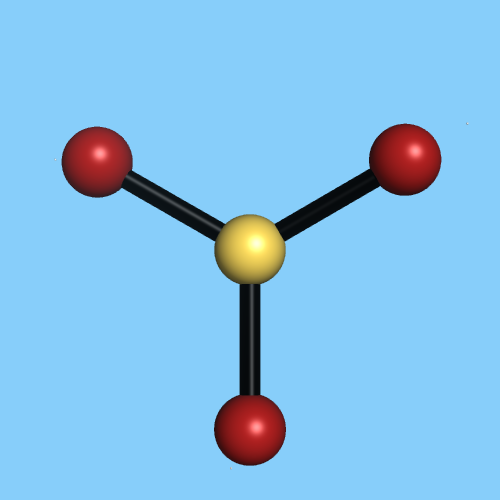}

\end{center}
\caption[$E_6^+$]{An alternating colouring of the $D_4$ graph that is used in the construction of the Coxeter plane. }
\label{figE8CoxPl}
\end{figure}

The  construction of the Coxeter plane is usually via a two-fold colouring of the  Dynkin diagram as illustrated for $D_4$ in Fig. \ref{figE8CoxPl}.
Since any finite Coxeter group has a tree-like diagram, one can partition the simple roots into two coloured sets (red and gold in the figure, or more commonly white and black) of roots, which are mutually orthogonal within each set. Since the Cartan matrix is positive definite, one eigenvector has  all positive entries, and is called the Perron-Frobenius eigenvector. For the example of $D_4$ here the Cartan matrix is 
$$\left( \begin{array}{cccc}
2 & 0 & 0 & -1\\
0 & 2 & 0 & -1\\
0 & 0 & 2 & -1\\
-1 & -1 & -1 & 2\end{array} \right),$$
 which has Perron-Frobenius eigenvector $(1, 1, 1, 2\cos \frac{\pi}{6})^T$. This allows one to show the existence of the invariant Coxeter plane. One takes the reciprocals (weights) of the simple roots and then defines two distinguished vectors: a white vector that is a linear combination of the white weights with the corresponding coefficients from the Perron-Frobenius eigenvector, and a black one, which is a linear combination of the black weights with the right Perron-Frobenius coefficients. The Coxeter plane is then the plane defined by these two vectors. In a Clifford algebra setup this is just given by the bivector that is  the outer product of these two vectors.

We start with a toy model in the plane for $I_2(n)$.
The simple roots for $I_2(n)$  can  be taken as
 $\alpha_1=e_1$ and $\alpha_2=-\cos{\frac{\pi}{n}}e_1+\sin{\frac{\pi}{n}}e_2$.
The Coxeter element $w$ is doubly covered in the Clifford setting by  the Coxeter versor $W$ acting as 	$v\rightarrow wv=\tilde{W}vW$. It describes the $n$-fold rotation encoded by the $I_2(n)$ Coxeter element and is given by 
\begin{equation}
W=\alpha_1\alpha_2=-\cos{\frac{\pi}{n}}+\sin{\frac{\pi}{n}}e_1e_2=-\exp{\left(-{\frac{\pi}{n} e_1e_2}\right)}.
\end{equation}
As one can see, different choices of signs and orders result in different signs in this formula but lead to an essentially equivalent descripion. In Clifford algebra it therefore immediately follows that the complex structure $i$ is actually the bivector $e_1e_2$ describing the plane (it could of course be nothing else) and that the action of the $I_2(n)$ Coxeter element is described by a versor that encodes rotations in this Coxeter plane. It  yields $h=n$ since trivially $W^n=(-1)^{n+1}$ such that $v\rightarrow wv=\tilde{W}vW=v$. Since $I=e_1e_2$ is the bivector defining the plane of $e_1$ and $e_2$, it anticommutes with both $e_1$ and $e_2$, since orthogonal vectors anticommute whilst parallel vectors commute. Therefore 
 one can take $W$ through to the left, which introduces a minus sign in  the bivector part and is thus equivalent to reversal $\tilde{W}$. Thus one arrives at the complex eigenvector equation 
$$
	v\rightarrow wv=\tilde{W}vW=\tilde{W}^2v=\exp{\left(\pm{2\pi I/n}\right)}v.
$$%
The standard result for the complex eigenvalues thus arises geometrically, without the need for complexification, and the complex structure instead arises from the bivector describing the rotation plane.  

Generally, if $W$ is a bivector exponential describing a rotation in a plane, then if a vector $v$ lies in this plane then the above applies and one can only pull $W$ to the left at the expense of reversing $W$ 
$$	v\rightarrow wv=\tilde{W}vW=\tilde{W}^2v,
$$%
giving rise to the complex eigenvalue equation. Conversely, if $v$ is orthogonal to that plane such that the bivector describing the plane commutes with $v$, one just has
$$	v\rightarrow wv=\tilde{W}vW=\tilde{W}Wv=v,
$$%
and thus the vector is invariant under that rotation. Thus, if a Coxeter versor $W$ factorises into orthogonal eigenspaces $W=W_1\dots W_k$  with the $W_i$s given by bivector exponentials and  $v$ lying in the plane defined by $W_j$, then all the orthogonal $W_i$s commute through and cancel out, whilst the one that defines the eigenplane that $v$ lies in  ($W_j$ described by the bivector $B_j$) gives the  complex eigenvalue equation with respect to $B_j$
$$\tilde{W}vW=\tilde{W_1}\dots\tilde{W_k}vW_k\dots W_1=\tilde{W_j}^2\dots\tilde{W_k}W_kv=\tilde{W_j}^2v=\exp(2\pi B_jm/h)v.
$$%
If $m$ is an exponent then so is $h-m$ since $w^{-1}$ will act as $${W_j}^2v=\exp(-2\pi B_jm/h)v=\exp(2\pi B_j(h-m)/h)v$$ (in particular $1$ and $h-1$ are always exponents arising from the Coxeter plane). As stated above these are just righthanded and lefthanded rotations in the respective eigenplanes, with bivectors $B_j$ giving the complex structures. If $W$ has pure vector factors then these act as reflections and trivially yield the exponents $h/2$.

This  description in terms of Clifford algebra therefore yields much deeper geometric insight, whilst avoiding ungeometric and unmotivated complexification. One sees that the eigenvalues and eigenvectors are not so much eigenvectors with complex eigenvalues, but rather eigenplanes of the Coxeter element. The complex nature of the eigenvalue arises because bivector exponentials describe rotations in planes with the plane bivector acting as an imaginary unit. Like for the 2D groups, the 3D and 4D geometry is completely governed by the above 2D geometry in the Coxeter plane, since the remaining normal vector (3D) or bivector (4D) are trivially fixed.

We explain the cases of $A_4$ and $B_4$ for illustrative purposes here, working in the four-dimensional Clifford algebra generated by the four orthogonal unit vectors $e_1$, $e_2$, $e_3$ and $e_4$. We start with $A_4$, which has Coxeter number $h=5$ and exponents $(1,2,3,4)$.  
We take as the simple roots 
$$\alpha_1=\frac{1}{\sqrt{2}}(e_2-e_1), \,\,\,\alpha_2=\frac{1}{\sqrt{2}}(e_3-e_2),\,\,\, \alpha_3=\frac{1}{\sqrt{2}}(e_4-e_3) $$
$$\text{ and } \alpha_4=\frac{1}{{2}}(\tau e_1+\tau e_2 +\tau e_3 +(\tau-2)e_4).$$
The Perron-Frobenius eigenvector of the Cartan matrix is $$(1, \tau, \tau, 1)^T=(1, 2\cos\frac{\pi}{5},  2\cos\frac{\pi}{5}, 1)^T,$$ which gives as the two coloured vectors $e_3+e_4$ and $-e_1+e_2+e_3+(2\tau+1)e_4$. From these, the Coxeter plane unit bivector $B_C\propto -e_1e_3-e_1e_4+e_2e_3+e_2e_4-1/2(\tau-1)e_3e_4$ is constructed.   The Coxeter versor $W=\alpha_3\alpha_1\alpha_2\alpha_4$ is  $4W= 1-e_2e_3+e_1e_4+(\tau-1)(e_3e_4+e_2e_4-e_1e_3)-(\tau+1)e_1e_2-(2\tau-1)e_1e_2e_3e_4$. It is straightforward (if tedious) to show that $\tilde{W}B_CW=B_C$. Thus the Coxeter element stabilises the Coxeter plane, i.e. it is an eigenplane of the Coxeter element. This Coxeter element can be written as the product of bivector exponentials in the orthogonal planes given by $B_C$ and  $IB_C$ (where $I$ customarily denotes the pseudoscalar of the space, i.e. here $I=e_1e_2e_3e_4$). This yields the correct angles and exponents  $(1,2,3,4)$ purely algebraically from this factorisation
 $$W=\exp\left(\frac{\pi}{5}B_C\right)\exp\left(-\frac{2\pi}{5}IB_C\right).$$ Taking a different order in the product of simple roots in the Coxeter element introduces overall minus signs as well as minus signs in the exponentials, but this does not ultimately change the geometric description.

The  projection of the $20$ vertices into the Coxeter plane forms two concentric decagons. 
$A_4$ is unusual in that such a projection from 4D usually yields 4 concentric rings of $h$ points, but here only two rings of $2h$ points. In fact, it consists of two copies of $H_2$ (the decagon) with a relative factor of $\tau$, which is due to the two-fold degeneracy in the Perron-Frobenius eigenvector. This is similar to $E_8$ and $H_4$ (and likewise for $D_6$ and $H_3$), as is well known and explained e.g. in \cite{Dechant2016AGACSE}, since by removing four of the eight nodes one gets a diagram folding from $A_4$ to $H_2$. 
In the Coxeter plane the Coxeter element therefore acts as a rotation by $2\pi/5$ , whilst in the plane $IB_C$ it acts as a rotation by $4\pi/5$. 

$B_4$  has exponents $(1, 3, 5, 7)$, which are again given by the factorisation of the Coxeter element   $$W=\exp\left(-\frac{\pi}{8}B_C\right)\exp\left(\frac{3\pi}{8}IB_C\right).$$ 
To provide the detail, we use the following choice of (normalised) simple roots
$$\alpha_1=e_4, \,\,\,\alpha_2=\frac{1}{\sqrt{2}}(e_3-e_4),\,\,\, \alpha_3=\frac{1}{\sqrt{2}}(e_2-e_3) \text{ and } \alpha_4=\frac{1}{\sqrt{2}}(e_1-e_2).$$
The Perron-Frobenius eigenvector gives coloured vectors $2\cos\frac{\pi}{8}(e_1+e_2)+2\cos\frac{3\pi}{8}(e_3+e_4)$ and $2e_1+\sqrt{2}(e_2+e3)$. 
The Coxeter plane unit bivector is therefore $B_C\propto 2\cos\frac{3\pi}{8}(\sqrt{2}(e_1e_3-e_1e_2-e_2e_4-e_3e_4))+2(e_2e_3-e_1e_4)$. The Coxeter versor $W=\alpha_3\alpha_1\alpha_2\alpha_4$ is  $4W/\sqrt{2}=1 +e_2e_3-e_1e_3+e_2e_4-e_1e_4+e_1e_2+e_3e_4+e_1e_2e_3e_4$. Using the form for $B_C$ found above this is actually equivalent to $W=\exp(-\frac{\pi}{8}B_C)\exp(\frac{3\pi}{8}IB_C)$. Therefore the standard complexification again misses that the eigenplanes and the correct exponents arise from this factorisation of the Coxeter element in the Clifford algebra.

The treatment for the other 4D groups is analogous. $A_4$ and $B_4$ were a good illustration of the Clifford Coxeter plane geometry irrespective of our induction construction; but we will discuss $D_4$, $F_4$ and $H_4$ in the context of the root systems $A_3$, $B_3$ and $H_3$ that induce them in the next section.  Table \ref{tab:2} shows the factorisation for those groups. The Coxeter element acts in the Coxeter plane $B_C$ as a rotation by $\pm 2\pi/h$ (clockwise and counterclockwise), and in the plane defined by $IB_C$ as $h$-fold rotations giving the remaining exponents algebraically.

\section{4D from 3D: the cases of the correspondence\label{sec_cases}}

{As we discussed in Section} \ref{indthrm}, each 3D root system induces a corresponding root system in 4D. We therefore now discuss each case in this correspondence, and find the factorisation of the 4D Coxeter element, which gives rise to the correct exponents. This therefore extends Arnold's original observation to include all cases of our new correspondence, as well as showing that a Clifford algebra approach to the geometry of the Coxeter plane is beneficial in many ways.

\subsection{$A_1^3$ and $A_1^4$}

The simplest root system $A_1$ is just a root and its negative, so three copies of $A_1$ are just given by three orthogonal unit vectors as the simple roots $\alpha_1=e_1, \alpha_2=e_2, \alpha_3=e_3$. Free multiplication (essentially amounting to multiplying together reflections in these simple roots) of these yields the eight elements in the 3D Clifford algebra from Section \ref{indthrm} and their negatives. Restricting to even products one gets $\pm 1, \pm e_1e_2, \pm e_2e_3, \pm e_3e_1$. These are essentially the quaternion group, and can be written as a collection of 4D vectors as $(\pm 1, 0, 0 ,0)$  { and permutations thereof}.  When thought of as a collection of 4D vectors, one sees that they are just the root system  $A_1^4$.

For simplicity when discussing the Coxeter element, rather than continuing to think of a 4D subspace of the 3D Clifford algebra, we switch to a formulation just in terms of the usual four Euclidean dimensions.  A choice of simple roots is given by  $\alpha_1=e_1$, $\alpha_2=e_2$, $\alpha_3=e_3$, $\alpha_4=e_4$ which gives $W=e_1e_2e_3e_4$. One could factor this into bivector exponentials in various (though somewhat trivial) ways. Since all four simple directions are orthogonal, reflections in pairs of them just give rotations by $\pi$ in the respective planes such that the exponents are all $\frac{h}{2}=\frac{2}{2}=1$ and $W=e_1e_2e_3e_4=(\cos \frac{\pi}{2}+\sin \frac{\pi}{2}e_1e_2)(\cos \frac{\pi}{2}+\sin \frac{\pi}{2}e_3e_4)=\exp(\frac{\pi}{2}e_1e_2)\exp(\frac{\pi}{2}e_3e_4)$ (though any other pair would give the same result). This gives exponents $1$ and $h-1=2-1=1$ from the first bivector exponential, and then the same again from the second. This is of course a special case of $A_1\times I_2(n)$, which we will revisit later.

\subsection{$A_3$ and $D_4$}

Starting with the tetrahedral root system $A_3$ by multiplying the simple roots, e.g. given by 
$$\alpha_1=\frac{1}{\sqrt{2}}(e_2-e_1),\,\,\, \alpha_2=\frac{1}{\sqrt{2}}(e_3-e_2)  \text{ and }  \alpha_3=\frac{1}{\sqrt{2}}(e_1+e_2),$$
one gets a group of 24 even products. This is the binary tetrahedral group consisting of 8 elements of the form $(\pm 1, 0, 0 ,0)$  and 16 of the form $\frac{1}{2}(\pm 1,\pm 1,\pm 1,\pm 1)$. As a collection of 4D vectors they form the $D_4$ root system.

{$D_4$} has {exponents $(1, 3, 3, 5)$} which is reflected in the fact that the Coxeter versor can be written as 
$W=\exp(-\frac{\pi}{6}B_C)\exp(\frac{\pi}{2}IB_C)=\exp(-\frac{\pi}{6}B_C)IB_C$. For instance  for simple roots
$$\alpha_1=e_1, \,\,\,\alpha_2=e_2,\,\,\, \alpha_3=e_3 \text{ and } \alpha_4=\frac{1}{2}(e_4-e_1-e_2-e_3)$$
the Coxeter element is $2W=2\alpha_1\alpha_2\alpha_3\alpha_4=e_1e_2e_3e_4-e_2e_3-e_1e_2+e_1e_3$. The Perron-Frobenius eigenvector is $(1, 1, 1, 2\cos \frac{\pi}{6})^T$, giving coloured vectors $e_1+e_2+e_3+3e_4$ and $4\cos \frac{\pi}{6} e_4$. The Coxeter plane bivector is therefore $1/\sqrt{3}(e_1+e_2+e_3)e_4$ given that $2\cos \frac{\pi}{6}=\sqrt{3}$. The Coxeter element thus factorises as $W=\exp(-\frac{\pi}{6}B_C)\exp(\frac{\pi}{2}IB_C)=\exp(-\frac{\pi}{6}B_C)IB_C$. The fact that the angle in the second bivector exponential is $\frac{\pi}{2}$ means that this second part of the Coxeter versor is simply the 
 product of two orthogonal vectors $e_1+e_2-2e_3$ and $e_1-e_2$ 
 rather than a `genuine' bivector exponential. The exponent is thus trivially $m=h/2=3$, as always for reflections. It can thus be easily checked that $\exp(-\frac{\pi}{6}B_C)\exp(\frac{\pi}{2}IB_C)=\frac{1}{2}(\sqrt{3}-B_C)IB_C=e_1e_2e_3e_4-e_2e_3-e_1e_2+e_1e_3$. Thus the Clifford factorisation gives rise to the correct exponents $(1, 3, 3, 5)$. Table \ref{tab:2} summarises the factorisations of the 4D Coxeter versors.

\subsection{$B_3$ and $F_4$}

As stated in Section \ref{indthrm}, the octahedral root system $B_3$, e.g. with a choice of simple roots
$$\alpha_1=e_3, \,\,\,\alpha_2=\frac{1}{\sqrt{2}}(e_2-e_3) \text{ and } \alpha_3=\frac{1}{\sqrt{2}}(e_1-e_2),$$
yields the root system $F_4$ via even products of roots, which form the binary octahedral group of order $48$.
They include the 24 spinors of the preceding subsection together with the 24 `dual' ones of the form $\frac{1}{\sqrt{2}}(\pm 1,\pm 1,0,0)$.

{$F_4$} has {exponents $(1, 5, 7, 11)$} which again is evident from the Clifford factorisation of the Coxeter element/versor $W=\exp(-\frac{\pi}{12}B_C)\exp(\frac{5\pi}{12}IB_C)$. 
We will again provide further detail for the choice of simple roots given by 
$$\alpha_1=\frac{1}{2}(e_4-e_1-e_2-e_3), \,\,\,\alpha_2=e_3,\,\,\, \alpha_3=\frac{1}{\sqrt{2}}(e_2-e_3) \text{ and } \alpha_4=\frac{1}{\sqrt{2}}(e_1-e_2).$$
The Perron-Frobenius eigenvector $(1, 2\cos \frac{\pi}{12}, 2\cos \frac{\pi}{12}, 1)^T$ gives coloured vectors
 $ 2\sqrt{2}\cos\frac{\pi}{12}(e_1+e_2)+(2+4\sqrt{2}\cos\frac{\pi}{12})e_4$ 
and 
$(2\cos\frac{\pi}{12}+\sqrt{2})e_1+2\cos\frac{\pi}{12}(e_2+e_3)+(6\cos\frac{\pi}{12}+\sqrt{2})e_4$. 
The Coxeter plane unit bivector is therefore 
$B_C\propto (4\cos\frac{\pi}{12}+\sqrt{2})(e_1e_3+e_2e_3+e_2e_4-e_1e_4)-4\cos\frac{\pi}{12}e_1e_2-(2\sqrt{2}+12\cos\frac{\pi}{12})e_3e_4$.
The Coxeter versor $W=\alpha_3\alpha_1\alpha_2\alpha_4$ is  
$4W=1 +e_2e_3+e_1e_3+e_2e_4-e_1e_4+3e_1e_2+e_3e_4+e_1e_2e_3e_4$.
 Using the form for $B_C$ found above it can again easily (if tediously) be shown that this is actually equivalent to $W=\exp(-\frac{\pi}{12}B_C)\exp(\frac{5\pi}{12}IB_C)$, giving rise to the exponents as claimed.

\subsection{$H_3$ and $H_4$}\label{H4}

Finally, the icosahedral root system $H_3$ gives rise to $H_4$ in 4D. For the choice of simple roots
$$\alpha_1=e_2, \,\,\,\alpha_2=-\frac{1}{2}(\tau e_1+e_2+(\tau-1)e_3)  \text{ and } \alpha_3=e_3,$$
one gets $120$ spinors forming the binary icosahedral group doubly covering the $60$ rotations of $A_5$. There are $8$, $16$, and $96$ respectively of the forms $(\pm 1, 0, 0, 0)$, $\frac{1}{2}(\pm 1, \pm 1, \pm 1, \pm 1)$ and $\frac{1}{2}(0, \pm 1, \pm (1-\tau), \pm \tau)$, which form the $H_4$ root system.  

{$H_4$} has Coxeter number $h=30$, {exponents $(1, 11, 19, 29)$} and factorisation $W=\exp(-\frac{\pi}{30}B_C)\exp(-\frac{11\pi}{30}IB_C)$.
The exact expressions get very unwieldy because of the golden ratio and square roots (see the Appendix), so for simplicity we give the results numerically for the choice of simple roots given by 
$$\alpha_1=\frac{1}{2}(\tau e_1-e_2+(\tau -1) e_4), \,\,\alpha_2=e_2,\,\, \alpha_3=-\frac{1}{{2}}((\tau -1)e_1+e_2+\tau e_3),\,\, \alpha_4=e_3.$$ 
The Perron-Frobenius eigenvector $(1, 1.989, 2.956, 2.405)^T$ gives coloured vectors proportional to
 $2.956e_4-e_1$ 
and 
$-4.784e_1+e_2+14.15e_4+1.209e_3$. 
The Coxeter plane unit bivector is then  
$B_C=-0.604e_2e_4-0.73e_3e_4-0.204e_1e_2-0.247e_1e_3$.
The Coxeter versor $W=\alpha_3\alpha_1\alpha_2\alpha_4$ is  
$4W=\tau+(2\tau -1)e_1e_3+\sigma e_3e_4-\tau^2 e_1e_2+e_2e_4-\sigma^2 e_1e_2e_3e_4$.
 Using the $B_C$ found above this is again equivalent to $W=\exp(-\frac{\pi}{30}B_C)\exp(-\frac{11\pi}{30}IB_C)$, giving the correct exponents.

\subsection{$A_1\times I_2(n)$ and $I_2(n)\times I_2(n)$}

The two copies of $I_2(n)$ are orthogonal in $I_2(n)\times I_2(n)$, such that one gets two sets of $1$ and $h-1=n-1$, i.e. $(1, 1, n-1, n-1)$. For instance  take $\alpha_1=e_1$, $\alpha_2=-\cos{\frac{\pi}{n}}e_1+\sin{\frac{\pi}{n}}e_2$, $\alpha_3=e_3$, $\alpha_4=-\cos{\frac{\pi}{n}}e_3+\sin{\frac{\pi}{n}}e_4$ as simple roots, then 
$W=(-\exp{\left(-\frac{\pi}{n} e_1e_2\right)})(-\exp{\left(-\frac{\pi}{n} e_3e_4\right)})=W_{12}W_{34}$ -- just two copies of the 2D case. 

This matches the Springer cone decomposition. For $I_2(n)$, the two bounding walls of a given Weyl chamber decompose the plane into a cone containing this Weyl chamber together with  its backward cone that also contains one Weyl chamber, whilst the complementary cones contain the other $n-1$ Weyl chambers each. In the projective plane one therefore gets the decomposition $2n=2(1+(n-1))$. For  $A_1\times I_2(n)$, one simply gets a doubling of this, since the $A_1$ just creates two copies of the $I_2(n)$ decomposition. One therefore gets the decomposition $4n=2(1+(n-1)+1+(n-1))$. 

This decomposition therefore matches the correct exponents of the Coxeter element even for the countably infinite family in the root system correspondence. Arnold's original link therefore extends to the full correspondence $(A_1\times I_2(n), A_3, B_3, H_3)\rightarrow (I_2(n)\times I_2(n), D_4, F_4, H_4)$ between root systems. 

\begin{table}

%
%
\begin{tabular}{|c|c|c|c|}

	\hline
&$h$&$m_i$&$W$
	\tabularnewline 	\hline 	\hline 
	$A_4$&$5$&$1,2,3,4$&$\exp\left(\frac{\pi}{{5}}B_C\right)\exp\left(-\frac{{2}\pi}{5}IB_C\right)$\textcolor{white}{\Huge{$H_o$}}
	\tabularnewline	\hline
	$B_4$&$8$&$1,3,5,7$&$\exp\left(-\frac{\pi}{{8}}B_C\right)\exp\left(\frac{{3}\pi}{8}IB_C\right)$\textcolor{white}{\Huge{$H_o$}}
	\tabularnewline
		\hline	\hline
	$D_4$&$6$&$1,3,3,5$&$\exp\left(-\frac{\pi}{{6}}B_C\right)\exp\left(\frac{\pi}{{2}}IB_C\right)$\textcolor{white}{\Huge{$H_o$}}
		\tabularnewline \hline	
		$F_4$&$12$&$1,5,7,11$&$\exp\left(-\frac{\pi}{{12}}B_C\right)\exp\left(\frac{{5}\pi}{12}IB_C\right)$\textcolor{white}{\Huge{$H_o$}}
			\tabularnewline
				\hline
	$H_4$&$30$&$1,11,19,29$&$\exp\left(-\frac{\pi}{{30}}B_C\right)\exp\left(-\frac{{11}\pi}{30}IB_C\right)$\textcolor{white}{\Huge{$H_o$}}
				\tabularnewline
	\hline
	\end{tabular}
	\caption{Clifford factorisations of the 4D Coxeter versors.}
	\label{tab:2}       
\end{table}

\section{Binary polyhedral groups, 4D root systems and the McKay correspondence}\label{sec_McKay}

The binary polyhedral groups are famously in a correspondence with the $ADE$-type affine Lie algebras via the McKay correspondence. Here we essentially have something between a Trinity and the McKay correspondence, since so far there is only one countably infinite family in the root system correspondence. However, the Trinity $(A_3, B_3, H_3)$ is connected in our correspondence to the binary tetrahedral, octahedral and icosahedral groups $(2T, 2O, 2I)$ and the 4D root systems that they induce $(D_4, F_4, H_4)$. However, in the McKay correspondence these are linked to $(E_6, E_7, E_8)$ via the tensor product structure of the binary groups and the Coxeter numbers $(12, 18, 30)$. We note that this is the number of roots in $(A_3, B_3, H_3)$, which hints at a direct correspondence between 3D root systems and $ADE$ Lie algebras, which we will elaborate on later. 

The root system $I_2(n)\times I_2(n)$ is essentially the dicyclic (or binary dihedral) group, which in the McKay correspondence is connected with the $D_n$ series of Lie algebras. If we want to extend the root system correspondence to a full $ADE$ correspondence we therefore need another countable family corresponding to $A_n$. This is in fact given by the 2D root systems $I_2(n)$. We had not previously included these in the correspondence, since 2D root systems are self-dual \cite{Dechant2012Induction} and thus did not appear to give new results: 

The Clifford algebra of two orthogonal unit vectors $e_1, \, e_2$ is 4-dimensional, 
$$
  \underbrace{\{1\}}_{\text{1 scalar}} \,\,\ \,\,\,\underbrace{\{e_1, e_2\}}_{\text{2 vectors}}  \,\,\, \,\,\, \underbrace{\{ e_1e_2\}}_{\text{1 bivector}},
$$
with the spaces of vectors and spinors both being of dimension two. There is therefore a canonical mapping between vectors $\alpha_i=a_1e_1+a_2e_2$ and spinors $R=a_1+a_2e_1e_2=: a_1+a_2I=e_1\alpha_i$ via multiplication e.g. with $e_1$, which is a bijection. Thus, for instance, for the choice of simple roots $\alpha_1=e_1$ and $\alpha_2=-\cos{\frac{\pi}{n}}e_1+\sin{\frac{\pi}{n}}e_2$ one sees that taking the spinors of this root system is essentially multiplication by $e_1$ and therefore the whole root system $I_2(n)$ gets dualised to itself. 
	The space of spinors has a natural Euclidean structure given by $R\tilde{R}=a_1^2+a_2^2$, i.e. a two-dimensional Euclidean vector space. This induces a rank-2 root system from any rank-2 root system in a similar way to the 3D-to-4D construction above, but is rather less interesting, as it does not yield any new root systems but just maps $I_2(n)$ to $I_2(n)$. However, as a complex/quaternionic group, $I_2(n)$ is precisely the cyclic group of order $2n$, $C_{2n}$, which via the McKay correspondence is linked to the missing $A_n$ family. This inclusion of the self-dual 2D root systems $I_2(n)$ therefore completes the $ADE$ correspondence 
	$$\begin{array}{c}
	(I_2(n), A_1\times I_2(n), A_3, B_3, H_3)\rightarrow (I_2(n), I_2(n)\times I_2(n), D_4, F_4, H_4),\\
	\\
	\text{ and } (I_2(n), I_2(n)\times I_2(n), D_4, F_4, H_4)
	\rightarrow (A_n, D_n, E_6, E_7, E_8).\end{array}$$
	\begin{table}

	%
	\begin{tabular}{|c|c||c|c|c||c|c|}
	\hline
	2D/3D&$|\Phi|$&4D &$G$ &$\sum d_i$ & $ADE$\textcolor{white}{\Huge{$H_o$}}&$h$ 	\tabularnewline 	\hline \hline
	$I_2(n)$ &$2n$&$I_2(n)$ &$C_{2n}$& $2n$ & $\tilde{A}_{2n-1}$ \textcolor{white}{\Huge{$H_o$}}& $2n$	\tabularnewline 	\hline 
	$A_1\times I_2(n)$ &$2n+2$&$I_2(n)\times I_2(n)$ &$\Dic_n$& $2n+2$ & $\tilde{D}_{n+2}$\textcolor{white}{\Huge{$H_o$}} & $2(n+1)$	\tabularnewline 	\hline 
	$A_3$ &$12$&$D_4$ &$2T$& $12$ & $\tilde{E}_{6}$ \textcolor{white}{\Huge{$H_o$}}& $12$	\tabularnewline 	\hline 
	$B_3$ &$18$&$F_4$ &$2O$& $18$ & $\tilde{E}_{7}$ \textcolor{white}{\Huge{$H_o$}}& $18$	\tabularnewline 	\hline 
	$H_3$ &$30$&$H_4$ &$2I$& $30$ & $\tilde{E}_{8}$ \textcolor{white}{\Huge{$H_o$}}& $30$	\tabularnewline 	\hline 
	\end{tabular}

	\caption{The correspondence of Clifford spinor induced root systems in 2D/3D and 4D.  The 4D root systems are binary polyhedral groups related to the $ADE$-type affine Lie algebras via the McKay correspondence. This thus extends to a correspondence between 2D/3D root systems and  $ADE$-type root systems, where the Coxeter number $h$ of the $ADE$ Lie algebras and the sum of the dimensions of the irreducible representations of the binary polyhedral group $G$,  $\sum d_i$, are in fact given by the number of roots  $|\Phi|$ in the 2D/3D root systems.}
	\label{tab:4}       
	\end{table}

The remarkable fact is now that our earlier observation that the number of roots  $(12, 18, 30)$ in $(A_3, B_3, H_3)$ matches the Coxeter number of $(E_6, E_7, E_8)$ actually extends to the whole $ADE$ correspondence: the number of roots of $I_2(n)$, $2n$, exactly matches the Coxeter number of  $\tilde{A}_{2n-1}$ that $C_{2n}$ corresponds to; likewise, the number of roots $2n+2$ in $A_1\times I_2(n)$ matches the Coxeter number $2(n+1)$ of the $\tilde{D}_{n+2}$ family that the dicyclic group $\Dic_n$ is in correspondence with. These results are summarised in Table \ref{tab:4}. The spinor-induced root systems $(I_2(n), A_1\times I_2(n), A_3, B_3, H_3)$ are therefore in correspondence with the $ADE$ Lie algebras via the intermediate polyhedral groups and the McKay correspondence, with $(2n, 2n+2, 12, 18, 30)$  simultaneously being the number of roots in the Platonic root systems, the sum of the dimensions of irreducible representations of the binary polyhedral groups and the ADE Coxeter numbers. In fact, $(I_2(12), I_2(18), I_2(30))$ were just found to be the exceptions within the $I_2(n)$ family in a very different context \cite{Kildetoft2018simple}. 

The next question therefore is whether there is a more direct correspondence between $(I_2(n), A_1\times I_2(n), A_3, B_3, H_3)$ and $ADE$ without involving the intermediate step. Furthermore, not all Lie algebras of $A$-type are actually included via this intermediate step. This is because root systems are always even, whilst the odd order cyclic groups also correspond to $A$-type Lie algebras which are therefore not covered. Including the 2D root systems seemed somewhat arbitrary anyway, so one could argue that one could also consider the cyclic groups of order $n$,  $C_n$, since they are subgroups and implicitly contained. However, that would lose our root system based reasoning which proved fruitful with the 4D Clifford spinor induction. We therefore explore a direct correspondence between $(I_2(n), A_1\times I_2(n), A_3, B_3, H_3)$ and $ADE$ in the next section, circumventing the intermediate step via 4D.

\section{A Trinity of correspondences: 2D/3D root systems, spinor induced root systems and  $ADE$}\label{sec_ADE}

The last two sections made connections between the 2D and 3D root systems and their induced 4D root systems; and via considering them as binary polyhedral groups also with $ADE$ affine Lie algebras via the McKay correspondence. In this section we discuss a direct connection between our collection of root systems $(I_2(n), A_1\times I_2(n), A_3, B_3, H_3)$ and the $ADE$ Lie algebras, again extending the connection between two Trinities.

The Trinities in question are of course $(A_3, B_3, H_3)$ and $(E_6, E_7, E_8)$. Their number of roots and Coxeter numbers match, as discussed above, but there is another more immediate, if mysterious, connection: the Platonic solids $(A_3, B_3, H_3)$ have corresponding characteristic triples of orders of rotations $((2,3,3), (2,3,4), (2,3,5))$ generated by pairs of generators of the Coxeter groups, or via the angles between the simple roots (as fractions of $\pi$). For instance  the icosahedral group has 2-, 3- and 5-fold rotations. The $E$-type diagrams also encode triples, albeit in a less obvious way: all three diagrams can be considered as consisting of three legs starting from a central node. The number of nodes in each leg also gives a number such that the three legs together give a triple. Thus, the $E_8$ diagram leads to the triple $(2,3,5)$; likewise $E_6$ gives $(2,3,3)$ whilst $E_7$ gives $(2,3,4)$ such that both Trinities give rise to the same set of triples. This connection seems very puzzling but is well-known (see, e.g. \cite{he2015sporadic}) and intuitive enough. It would be interesting to find a construction that makes this elusive link explicit, akin to the simple connection between 3D and 4D root systems. The interesting question is now whether this observation also extends to our full correspondence between $(I_2(n), A_1\times I_2(n), A_3, B_3, H_3)$ and the $ADE$ Lie algebras. Which diagrams would the other root systems, i.e. the two countable families, correspond to?

The product of the two simple roots in $I_2(n)$ simply gives  rise to an $n$-fold rotation. This corresponds to a single leg with $n$ nodes in the diagram as a simple string and nothing else. This is exactly the $A_n$ diagram. This time, all $A_n$ are achieved from the $I_2(n)$, in contrast to the initial connection with the McKay correspondence above. 
The root systems $A_1\times I_2(n)$ encode a triple of rotation orders: the two simple roots of $I_2(n)$ still give a string of length $n$, but now the simple root of $A_1$ with either of the two simple roots of $I_2(n)$ will give rise to a 2-fold rotation, as they are orthogonal. The triple is therefore  $(2,2,n)$, which gives a diagram with one leg of length $n$ which meets two legs of length $2$, i.e. it is of $D$-type. It is in fact $D_{n+2}$, which was to be expected from the connection with the McKay correspondence. This correspondence between rotation orders and lengths of legs in Dynkin diagrams therefore extends to the full set of root systems $(I_2(n), A_1\times I_2(n), A_3, B_3, H_3)$ that we have established in this paper and the $ADE$ Lie algebras. A tangible construction connecting the two explicitly akin to the above construction in terms of Clifford spinors would be desirable. Our recent construction of the $E_8$ root system from the $H_3$ root system in a related Clifford construction  \cite{Dechant2016Birth} perhaps hints that this is again a correspondence between root systems, rather than operating at the level of the  Lie algebras -- after all, the Lie algebras and groups are defined by their respective root systems (up to integrability conditions which are met by crystallographic root systems). Table \ref{tab:3} summarises this correspondence. 

\begin{table}

\label{tab_5}       
%
%
\begin{tabular}{|c|c|c||c|c|c|}
\hline
2D/3D& &rot &$ADE$ & & legs  	\tabularnewline 	\hline \hline
$I_2(n)$ &\begin{tikzpicture}[scale=0.4,
knoten/.style={        circle,      inner sep=.10cm,        draw}
]
\node at  (1,.7) (knoten1) [knoten,   color=white!0!black, ball color=white ] {};
\node at  (3,.7) (knoten2) [knoten,   color=white!0!black, ball color=white ] {};

\node at  (-1,0) (alpha0)  {};
\node at  (1,0)  (alpha1) {};
\node at  (2,1.2)  (alpha2) {{$n$}};

\path  (knoten1) edge (knoten2);
 
\end{tikzpicture}&$n$ &$A_n$& \begin{tikzpicture}[scale=0.4,
knoten/.style={        circle,      inner sep=.10cm,        draw}
]
\node at  (5,.7) (knoten3) [knoten,   color=white!0!black, ball color=white ] {};
\node at  (7,.7) (knoten4) [knoten,   color=white!0!black, ball color=white ] {};
\node at (9,.7) (knoten6) {$\dots$};

\node at  (11,.7) (knoten7) [knoten,   color=white!0!black, ball color=white ] {};
\node at  (13,.7) (knoten8) [knoten,   color=white!0!black, ball color=white ] {};

\node at (8,.7) (kknoten4) {};
\node at (10,.7) (kknoten7) {};

\path  (knoten3) edge (knoten4);
\path  (knoten4) edge (kknoten4);
\path  (kknoten7) edge (knoten7);
\path  (knoten7) edge (knoten8);

\end{tikzpicture} & $n$	\tabularnewline 	\hline 
$A_1\times I_2(n)$ &\begin{tikzpicture}[scale=0.4,
knoten/.style={        circle,      inner sep=.10cm,        draw}
]
\node at (-1,.7) (knoten0) [knoten,   color=white!0!black, ball color=white ] {};
\node at  (1,.7) (knoten1) [knoten,   color=white!0!black, ball color=white ] {};
\node at  (3,.7) (knoten2) [knoten,   color=white!0!black, ball color=white ] {};

\node at  (-1,0) (alpha0)  {};
\node at  (1,0)  (alpha1) {};
\node at  (2,1.2)  (alpha2) {{$n$}};

\path  (knoten1) edge (knoten2);
 
\end{tikzpicture} &$2,2,n$& ${D}_{n+2}$ & \begin{tikzpicture}[scale=0.4,
knoten/.style={        circle,      inner sep=.10cm,        draw}
]
\node at  (5.3,-0.3) (knoten3a) [knoten,   color=white!0!black, ball color=white ] {};
\node at  (5.3,1.7) (knoten3b) [knoten,   color=white!0!black, ball color=white ] {};
\node at  (7,.7) (knoten4) [knoten,   color=white!0!black, ball color=white ] {};
\node at (9,.7) (knoten6) {$\dots$};

\node at  (11,.7) (knoten7) [knoten,   color=white!0!black, ball color=white ] {};
\node at  (13,.7) (knoten8) [knoten,   color=white!0!black, ball color=white ] {};

\node at (8,.7) (kknoten4) {};
\node at (10,.7) (kknoten7) {};

\path  (knoten3a) edge (knoten4);
\path  (knoten3b) edge (knoten4);
\path  (knoten4) edge (kknoten4);
\path  (kknoten7) edge (knoten7);
\path  (knoten7) edge (knoten8);

\end{tikzpicture}& $2,2,n$	\tabularnewline 	\hline 
$A_3$ &\begin{tikzpicture}[scale=0.4,
knoten/.style={        circle,      inner sep=.10cm,        draw}
]
\node at (-1,.7) (knoten0) [knoten,   color=white!0!black, ball color=white ] {};
\node at  (1,.7) (knoten1) [knoten,   color=white!0!black, ball color=white ] {};
\node at  (3,.7) (knoten2) [knoten,   color=white!0!black, ball color=white ] {};

\node at  (-1,0) (alpha0)  {};
\node at  (1,0)  (alpha1) {};

\path  (knoten0) edge (knoten1) ;
\path  (knoten1) edge (knoten2);
 
\end{tikzpicture}&$2,3,3$ & ${E}_{6}$ & \begin{tikzpicture}[scale=0.4,
knoten/.style={        circle,      inner sep=.10cm,        draw}
]
\node at  (5,.7) (knoten3) [knoten,   color=white!0!black, ball color=white ] {};
\node at  (7,.7) (knoten4) [knoten,   color=white!0!black, ball color=white ] {};
\node at  (9,.7) (knoten6) [knoten,   color=white!0!black, ball color=white ] {};
\node at  (11,.7) (knoten7) [knoten,   color=white!0!black, ball color=white ] {};
\node at  (13,.7) (knoten8) [knoten,   color=white!0!black, ball color=white ] {};
\node at  (9,2.7) (knoten9) [knoten,   color=white!0!black, ball color=white ] {};

\path  (knoten3) edge (knoten4);
\path  (knoten4) edge (knoten6);
\path  (knoten6) edge (knoten9);
\path  (knoten6) edge (knoten7);
\path  (knoten7) edge (knoten8);

\end{tikzpicture}&$2,3,3$	\tabularnewline 	\hline 
$B_3$ &\begin{tikzpicture}[scale=0.4,
knoten/.style={        circle,      inner sep=.10cm,        draw}
]
\node at (-1,.7) (knoten0) [knoten,   color=white!0!black, ball color=white ] {};
\node at  (1,.7) (knoten1) [knoten,   color=white!0!black, ball color=white ] {};
\node at  (3,.7) (knoten2) [knoten,   color=white!0!black, ball color=white ] {};

\node at  (-1,0) (alpha0)  {};
\node at  (1,0)  (alpha1) {};
\node at  (2,1.2)  (alpha2) {{$4$}};

\path  (knoten0) edge (knoten1) ;
\path  (knoten1) edge (knoten2);
 
\end{tikzpicture}&$2,3,4$ &${E}_{7}$ & \begin{tikzpicture}[scale=0.4,
knoten/.style={        circle,      inner sep=.10cm,        draw}
]
\node at  (3,.7) (knoten2) [knoten,   color=white!0!black, ball color=white ] {};
\node at  (5,.7) (knoten3) [knoten,   color=white!0!black, ball color=white ] {};
\node at  (7,.7) (knoten4) [knoten,   color=white!0!black, ball color=white ] {};
\node at  (9,.7) (knoten6) [knoten,   color=white!0!black, ball color=white ] {};
\node at  (11,.7) (knoten7) [knoten,   color=white!0!black, ball color=white ] {};
\node at  (13,.7) (knoten8) [knoten,   color=white!0!black, ball color=white ] {};
\node at  (9,2.7) (knoten9) [knoten,   color=white!0!black, ball color=white ] {};

\path  (knoten2) edge (knoten3);
\path  (knoten3) edge (knoten4);
\path  (knoten4) edge (knoten6);
\path  (knoten6) edge (knoten9);
\path  (knoten6) edge (knoten7);
\path  (knoten7) edge (knoten8);

\end{tikzpicture}&  $2,3,4$	\tabularnewline 	\hline 
$H_3$ &\begin{tikzpicture}[scale=0.4,
knoten/.style={        circle,      inner sep=.10cm,        draw}
]
\node at (-1,.7) (knoten0) [knoten,   color=white!0!black, ball color=white ] {};
\node at  (1,.7) (knoten1) [knoten,   color=white!0!black, ball color=white ] {};
\node at  (3,.7) (knoten2) [knoten,   color=white!0!black, ball color=white ] {};

\node at  (-1,0) (alpha0)  {};
\node at  (1,0)  (alpha1) {};
\node at  (2,1.2)  (alpha2) {{$5$}};

\path  (knoten0) edge (knoten1) ;
\path  (knoten1) edge (knoten2);
 
\end{tikzpicture}&$2,3,5$ &${E}_{8}$ & \begin{tikzpicture}[scale=0.4,
knoten/.style={        circle,      inner sep=.10cm,        draw}
]
\node at  (1,.7) (knoten1) [knoten,   color=white!0!black, ball color=white ] {};
\node at  (3,.7) (knoten2) [knoten,   color=white!0!black, ball color=white ] {};
\node at  (5,.7) (knoten3) [knoten,   color=white!0!black, ball color=white ] {};
\node at  (7,.7) (knoten4) [knoten,   color=white!0!black, ball color=white ] {};
\node at  (9,.7) (knoten6) [knoten,   color=white!0!black, ball color=white ] {};
\node at  (11,.7) (knoten7) [knoten,   color=white!0!black, ball color=white ] {};
\node at  (13,.7) (knoten8) [knoten,   color=white!0!black, ball color=white ] {};
\node at  (9,2.7) (knoten9) [knoten,   color=white!0!black, ball color=white ] {};

\path  (knoten1) edge (knoten2);
\path  (knoten2) edge (knoten3);
\path  (knoten3) edge (knoten4);
\path  (knoten4) edge (knoten6);
\path  (knoten6) edge (knoten9);
\path  (knoten6) edge (knoten7);
\path  (knoten7) edge (knoten8);

\end{tikzpicture}&  $2,3,5$	\tabularnewline 	\hline 
\end{tabular}
\caption{Extending the Trinity/root system correspondence to a correspondence between  2D/3D and $ADE$ root systems, omitting the intermediate step via 4D root systems/binary polyhedral groups: the 2D/3D root systems generate rotations of orders given by the angle between simple roots. The rotation orders from the 2D/3D root systems are in a one-to-one correspondence with the lengths of the legs in the $ADE$ Coxeter-Dynkin diagrams. }

\label{tab:3} 
\end{table}

Thus, there are three interrelated classes of objects: the 2D/3D root systems, the 2D/4D induced root systems, and the $ADE$-type root systems. The McKay correspondence was a rather mysterious correspondence between the latter two, whilst our Clifford construction straightforwardly related the first two; the connection between the first and the last is intuitively clear, if not explicitly. Thus, the intermediate step via the McKay correspondence was probably not  the most natural way to think about these interrelations. It appears that rather, the three classes of objects are related more akin to the $D_4$ diagram in Fig. \ref{figE8CoxPl}, where all three classes of objects have relations to each other, and perhaps to a central, not yet identified, concept. In fact, our two new correspondences seem very straightforward; perhaps the reason the McKay correspondence seemed so remarkable was because in a sense it was the least natural out of the connections between the three classes of objects.

\section{Conclusions}\label{sec_concl}
In this paper, we have investigated whether Arnold's original observation linking the two Trinities $(A_3, B_3, H_3)$ and $(D_4, F_4, H_4)$ via Weyl chamber decomposition and exponents in fact extends to  the additional cases included in our recent correspondence between 3D and 4D root systems $(A_1\times I_2(n), A_3, B_3, H_3)$ and $(I_2(n)\times I_2(n), D_4, F_4, H_4)$. We discussed the different cases of the 4D Coxeter plane in a Clifford algebra framework and showed how the correct exponents arise from the factorisation of the Coxeter element. Arnold's link thus indeed extends to this new correspondence. An extension of the general root system construction might connect the Weyl chamber decomposition to the 4D exponents in generality. This proof will be the subject of future work.  

Since our Clifford construction gives rise to the (binary) polyhedral groups, which are in McKay correspondence to the $ADE$ affine Lie algebras, we discussed whether the self-dual 2D root systems should be included in this correspondence, enlarging it to $(I_2(n), A_1\times I_2(n), A_3, B_3, H_3)$ and $(I_2(n), I_2(n)\times I_2(n), D_4, F_4, H_4)$. This connects these root systems with $ADE$ affine Lie algebras via the McKay correspondence as an intermediate. Throughout this new correspondence, we have the -- to our knowledge -- new observation that the number of roots in 2D/3D matches the Coxeter number of the corresponding $ADE$ Lie algebra, as well as the sum of the dimensions of the irreducible representations of the intermediate binary polyhedral group. Again there should be a general proof of why this should be the case, which will be the focus of further work. 

We finally discussed a direct connection between the 2D/3D root systems and the $ADE$ Lie algebras (not affine) via rotation orders and lengths of legs in Dynkin diagrams. The fact that we had come up with the set of 2D/3D root systems $(I_2(n), A_1\times I_2(n), A_3, B_3, H_3)$  for entirely different reasons makes the corresponding Dynkin diagrams a prediction. Again the link between two Trinities $(A_3, B_3, H_3)$ and $(E_6, E_7, E_8)$ was extended to the full correspondence; the additional families $I_2(n)$ and $A_1\times I_2(n)$ in fact exactly encode the (missing) $A_n$ and $D_n$ diagrams, which thus extends the correspondence to a full $ADE$ correspondence. This in many ways is more natural than the original McKay correspondence, though again a concrete construction linking the two is missing. However, from our earlier work, which is at the level of root systems, it appears that there is nothing that requires the correspondence to be at the level of the Lie algebra or group. Rather, it seems much more natural that this is again another correspondence between root systems: 2D/3D, 2D/4D and $ADE$-type root systems. We have explicitly constructed the $E_8$ root system from the $H_3$ root system in earlier work \cite{Dechant2016Birth}, and it seems likely that a related construction will apply to the full correspondence, making the link more explicit than operating at the level of the diagrams (akin to the original McKay correspondence). The McKay correspondence appears to be determined already by properties of the $ADE$ root systems such as Coxeter number and diagrams -- no properties of the Lie groups and algebras are needed as such. There might thus be a conceptual unification at the level of root systems -- for which we have argued Clifford algebras are a natural framework.

\subsection*{Ethics statement}
This work did not use personal data, nor does the research have ethical implications. The research was performed with the highest standards of academic integrity. 
\subsection*{Data accessibility}
This work does not have any experimental data.
\subsection*{Competing interests}
The author has no conflict of interest to declare.
\subsection*{Author's contributions}
The single author PPD conceived the project, performed all the calculations and wrote the paper. 
\subsection*{Acknowledgements}\label{ack}
This paper is dedicated to the memory of the late Lady Isabel Butterfield and Lord John Butterfield. I am grateful for their and  their son Jeremy Butterfield's support and friendship over the last 15 years, beginning with the Downing Lord Butterfield Harvard Award and henceforth. I would also like to thank Yang-Hui He, Peter Cameron, John McKay, Alastair King, Rob Wilson, Terry Gannon, John Baez, Jim Humphreys, Reidun Twarock, Anthony Lasenby, Joan Lasenby, David Hestenes and Eckhard Hitzer.
\subsection*{Funding statement}
This work was not funded by any specific grant.

\section*{Appendix}

This appendix gives the exact results for the $H_4$ calculation given numerically in Section \ref{H4}.

 The Cartan matrix for $H_4$ is
$$\left( \begin{array}{cccc}
2 & -1 & 0 & 0\\
-1 & 2 & -1 & 0\\
0 & -1 & 2 & -\tau\\
0 & 0 & -\tau & 2\end{array} \right),$$
whilst the inverse basis for the given choice of simple roots is $\omega_1=2\tau e_4$, $\omega_2=-\tau e_1+e_2+(3\tau+1)e_4$, $\omega_3 =-2\tau e_1+(4\tau+2)e_4$ and $\omega_4=-(1+\tau)e_1+e_3+(3\tau+2)e_4$.

The Perron-Frobenius eigenvector of this Cartan matrix is proportional to $$\left[ \begin {array}{c} 4+4\,\sqrt {5}\\ \noalign{\medskip}2\,\sqrt 
{7+\sqrt {30+6\,\sqrt {5}}+\sqrt {5}} \left( 1+\sqrt {5} \right) 
\\ \noalign{\medskip}\sqrt {30+6\,\sqrt {5}}\sqrt {5}+4\,\sqrt {5}+
\sqrt {30+6\,\sqrt {5}}+8\\ \noalign{\medskip} \left( \sqrt {5}+\sqrt 
{30+6\,\sqrt {5}}-1 \right) \sqrt {7+\sqrt {30+6\,\sqrt {5}}+\sqrt {5}
}\end {array} \right].
$$


This yields coloured vectors proportional to 
$$\left[ \begin {array}{c}  \left( -2\,\sqrt {5}-6 \right) \sqrt {30+6\,\sqrt {5}}-12\,
\sqrt {5}-28\\ \noalign{\medskip}0
\\ \noalign{\medskip}0\\ \noalign{\medskip} \left( 6\,\sqrt {5}+14 \right) \sqrt {30+6\,\sqrt {5}}+40\,
\sqrt {5}+96\end {array} \right] =:\left[ \begin {array}{c}  a_1\\ \noalign{\medskip}0
\\ \noalign{\medskip}0\\ \noalign{\medskip} a_4\end {array} \right]
$$
and
$$
\left[ \begin {array}{c} \left( \sqrt {5}+
3 \right) \sqrt {30+6\,\sqrt {5}}+6\,\sqrt {5}+14\\ \noalign{\medskip}-4-4\,\sqrt {5} 
\\ \noalign{\medskip}-2\,\sqrt {30+6\,\sqrt {5}}-2\,\sqrt {5}+2\\ \noalign{\medskip}  \left( -3\,\sqrt 
{5}-7 \right) \sqrt {30+6\,\sqrt {5}}-20\,\sqrt {5}-48\end {array} \right]=:\left[ \begin {array}{c} b_1\\ \noalign{\medskip}b_2 
\\ \noalign{\medskip}b_3\\ \noalign{\medskip}  b_4\end {array} \right]
$$
The Coxeter plane bivector is (up to normalisation) the outer product of those two coloured vectors and is hence of the form 
$$ a_1b_2e_1e_2+a_1b_3e_1e_3+(a_1b_4-b_1a_4)e_1e_4-b_2a_4e_2e_4-b_3a_4e_3e_4.$$ The term in brackets actually cancels, whilst the remaining terms are given in terms of the above coefficients. This gives the numerical results stated earlier up to proportionality, which gave the normalised (unit) Coxeter bivector. The bivector discussed here has norm squared $$\left( -1044480\,\sqrt {5}-2334720 \right) \sqrt {30+6\,\sqrt {5}}-
6893568\,\sqrt {5}-15421440.$$ The coefficients are divisible by $12288$, leaving $$\left( -85\,\sqrt {5}-190 \right) \sqrt {30+6\,\sqrt {5}}-
561\,\sqrt {5}-1255.$$ Two of the products simplify slightly: 
$$a_1b_3=\left( 32\,\sqrt {5}+64 \right) \sqrt {30+6\,\sqrt {5}}+224\,\sqrt {5
}+544 \text { and }$$ $$a_4b_3=\left( -96\,\sqrt {5}-224 \right) \sqrt {30+6\,\sqrt {5}}-640\,\sqrt 
{5}-1408.$$ 


\end{document}